\documentclass[journal]{IEEEtran}

\usepackage{epsfig,amsmath,amssymb,epsf,amsthm,scalefnt,multirow,subfig}
\usepackage{xcolor}
\usepackage{float}
\usepackage{cite}
\usepackage{psfrag}
\usepackage{algpseudocode}
\usepackage{algorithm}

\newtheorem*{lemma*}{Lemma}

\newtheorem{proposition}{Proposition}

  \def\cC{{\mathcal{C}}}

 \def\cN{{\mathcal{N}}}

\def\b0{{\pmb{0}}} 

\def\ba{{\mathbf{a}}}   
 \def\bff{{\mathbf{f}}}  \def\bh{{\mathbf{h}}}

  \def\bs{{\mathbf{s}}} 
\def\bu{{\mathbf{u}}} \def\bv{{\mathbf{v}}} \def\bw{{\mathbf{w}}} 
 \def\bz{{\mathbf{z}}}  

\def\bA{{\mathbf{A}}}   
 \def\bF{{\mathbf{F}}}  \def\bH{{\mathbf{H}}}
\def\bI{{\mathbf{I}}}  \def\bK{{\mathbf{K}}}

\makeatletter
\def\underbracex#1#2{\mathop{\vtop{\m@th\ialign{##\crcr
   $\hfil\displaystyle{#2}\hfil$\crcr
   \noalign{\kern3\p@\nointerlineskip}%
   #1\crcr\noalign{\kern3\p@}}}}\limits}

\def\underbracea{\underbracex\upbracefilla}

\def\upbracefilla{$\m@th \setbox\z@\hbox{$\braceld$}%
  \bracelu\leaders\vrule \@height\ht\z@ \@depth\z@\hfill 
\kern\p@\vrule \@width\p@\kern\p@\vrule \@width\p@\kern\p@\vrule \@width\p@
$}

\def\upbracefillb{$\m@th \setbox\z@\hbox{$\braceld$}%
\vrule \@width\p@\kern\p@\vrule \@width\p@\kern\p@\vrule \@width\p@\kern\p@
 \leaders\vrule \@height\ht\z@ \@depth\z@\hfill\bracerd
  \braceld\leaders\vrule \@height\ht\z@ \@depth\z@\hfill
\kern\p@\vrule \@width\p@\kern\p@\vrule \@width\p@\kern\p@\vrule \@width\p@
$}

\def\upbracefillc{$\m@th \setbox\z@\hbox{$\braceld$}%
\vrule \@width\p@\kern\p@\vrule \@width\p@\kern\p@\vrule \@width\p@\kern\p@
\leaders\vrule \@height\ht\z@ \@depth\z@\hfill
\kern\p@\vrule \@width\p@\kern\p@\vrule \@width\p@\kern\p@\vrule \@width\p@
$}

\def\underbraced{\underbracex\upbracefilld}
\def\upbracefilld{$\m@th \setbox\z@\hbox{$\braceld$}%
\vrule \@width\p@\kern\p@\vrule \@width\p@\kern\p@\vrule \@width\p@\kern\p@
 \leaders\vrule \@height\ht\z@ \@depth\z@\hfill\braceru$}

\makeatother

\begin{document}

\bstctlcite{IEEEexample:BSTcontrol}

\title{\huge Robust Transmission Design for Active RIS-Aided Systems}


\author{\IEEEauthorblockN{Jinho Yang, Hyeongtaek Lee, and Junil Choi}
    \thanks{This work was supported in part by LG Electronics Inc.; in part by Institute of Information \& communications Technology Planning \& Evaluation (IITP) grant funded by the Korea government(MSIT) (No. RS-2024-00395824, Development of Cloud virtualized RAN (vRAN) system supporting upper-midband); in part by Institute of Information \& communications Technology Planning \& Evaluation (IITP) under Open RAN Education and Training Program(IITP-2025-RS-2024-00429088) grant funded by the Korea government(MSIT); and in part by Institute of Information \& communications Technology Planning \& Evaluation (IITP) grant funded by the Korea government(MSIT) (No.2021-000269, Development of sub-THz band wireless transmission and access core technology for 6G Tbps data rate).}
    \thanks{Copyright (c) 2025 IEEE. Personal use of this material is permitted. However, permission to use this material for any other purposes must be obtained from the IEEE by sending a request to pubs-permissions@ieee.org.}
    \thanks{Jinho Yang and Junil Choi are with the School of Electrical Engineering, Korea Advanced Institute of Science and Technology, Daejeon 34141, South Korea (e-mail:\{dwplo3479; junil\}@kaist.ac.kr).}
    \thanks{Hyeongtaek Lee is with the Department of Electronic and Electrical Engineering, Ewha Womans University, Seoul 03760, South Korea (e-mail: htlee@ewha.ac.kr).}}


\maketitle

\begin{abstract}
Different from conventional passive reconfigurable intelligent surfaces (RISs), incident signals and thermal noise can be amplified at active RISs.
By exploiting the amplifying capability of active RISs, noticeable performance improvement can be expected when precise channel state information (CSI) is available.
Since obtaining perfect CSI related to an RIS is difficult in practice, a robust transmission design is proposed in this paper to tackle the channel uncertainty issue, which will be more severe for active RIS-aided systems.
To account for the worst-case scenario, the minimum achievable rate of each user is derived under a statistical CSI error model.
Subsequently, an optimization problem is formulated to maximize the sum of the minimum achievable rate.
Since the objective function is non-concave, the formulated problem is transformed into a tractable lower bound maximization problem, which is solved using an alternating optimization method.
Numerical results show that the proposed robust design outperforms a baseline scheme that only exploits estimated CSI.
\end{abstract}

\renewcommand\IEEEkeywordsname{Index Terms}
\begin{IEEEkeywords}
	Active reconfigurable intelligent surface (RIS), robust transmission design, statistical channel state information (CSI) error.
\end{IEEEkeywords}

\section{Introduction}

Reconfigurable intelligent surface (RIS) has recently attracted a lot of attention as a new technology that can increase spectral and energy efficiency at a low cost \cite{Intro1, Intro2, Intro3}.
Conventional RIS is a planar surface equipped with a large number of passive scattering elements that can control the phase of incident signals.
While signal propagation environments can be adjusted as desired by properly operating RIS, performance improvements may be limited due to the double fading attenuation \cite{Pathloss}.
To overcome this, utilizing an active RIS is recently proposed that each RIS element is constructed by integrating a reflection-type amplifier into the conventional passive reflective element.
With the addition of the amplifier, the active RIS can adjust the phase of incoming signals and even amplify their amplitudes \cite{Active1, Active2}.

By exploiting its amplification capability, extensive researches on active RIS-aided systems have been conducted. 
A sum-rate maximization problem was formulated in \cite{Active2} and solved by adopting alternating optimization (AO) and fractional programming methods.
In \cite{Active4}, theoretical performance comparison between active RIS-aided and passive RIS-aided systems was analyzed under the same power budget.
An energy-constrained wireless network aided by the active RIS was investigated with different multiple access schemes in \cite{Active5}.
Although the existing works made considerable performance improvements, they all assumed perfect channel state information (CSI) \cite{Active1, Active2, Active3, Active4, Active5}.
However, due to inevitable estimation error, obtaining accurate CSI is difficult in general.
Moreover, the signal amplification characteristic of the active RIS increases the impact of the channel mismatch, which can lead to significant performance loss.
Therefore, to exploit the active RIS, it is necessary to consider the CSI estimation error when designing a transmit beamformer at a base station (BS) and reflection coefficients at an RIS.

Although several works proposed robust designs in passive RIS-aided systems \cite{Robust1, Robust2, Robust3, Robust4, Robust5}, these designs cannot be directly applied to active RIS-aided systems because the noise amplification through the RIS and the power limitation at the RIS must be additionally taken into account.
To address this issue, a few recent works have developed novel robust transmission designs for active RIS-aided systems under the imperfect CSI assumption \cite{Active_Robust1, Active_Robust2, Active_Robust3}.
Both the average achievable rate maximization problem and the average transmit power minimization problem were investigated in \cite{Active_Robust1}.
The minimum signal-to-interference-plus-noise ratio was maximized also considering hardware impairments in \cite{Active_Robust2}.
In \cite{Active_Robust3}, energy efficiency was maximized by employing a deep reinforcement learning method.
However, the previous works did not consider the BS-user direct link itself or its estimation error, which should be taken into account in practice.

In this paper, a robust design for the transmit beamformer at the BS and the reflection coefficients at the RIS is developed.
To consider the worst-case, for the first time in active RIS-aided systems to the best of our knowledge, we derive the minimum achievable rate with a popular statistical CSI error model.
An optimization problem is formulated to maximize the sum of the minimum achievable rate.
Since the problem is non-convex, we reformulate it by replacing the objective function with a tractable lower bound and then apply the AO method.
Numerical results demonstrate that the proposed robust design achieves better sum-rate performance than a baseline non-robust case, which does not deal with the estimation error.

\textbf{Notation:} Lower and upper boldface letters denote column vectors and matrices.
$\bA^{\mathrm{T}}$ and $\bA^{\mathrm{H}}$ represent the transpose and the conjugate transpose of a matrix $\bA$.
$\mathrm{Tr}(\bA)$ denotes the trace of a square matrix $\bA$.
The $\ell_2$-norm of a vector $\ba$ and the Frobenius norm of a matrix $\bA$ are denoted by $\Vert \ba \Vert_2$ and $\Vert \bA \Vert_\mathrm{F}$.
The diagonal matrix whose diagonal elements are each entry of a vector $\ba$ is represented by $\mathrm{diag}(\ba)$.
The $m \times 1$ all zero-vector and the $m \times m$ identity matrix are represented by $\mathbf{0}_m$ and $\bI_m$.
$\cC\cN(\boldsymbol{\mu}, \bK)$ denotes the complex Gaussian distribution with mean $\boldsymbol{\mu}$ and covariance $\bK$.

\section{System Model}
As shown in Fig.~\ref{Fig:System_model}, we consider a multi-user multiple-input single-output (MU-MISO) downlink system aided by an active RIS.
The BS deploys $N$ antennas to serve $K$ single-antenna users.
The active RIS consists of $M$ reflective elements, each of which can independently reflect the incoming signal with phase shift and amplification. 
We assume the active RIS is connected to the BS through a control link.
\begin{figure}[t]
    \centering
    \includegraphics[width=0.8\columnwidth]{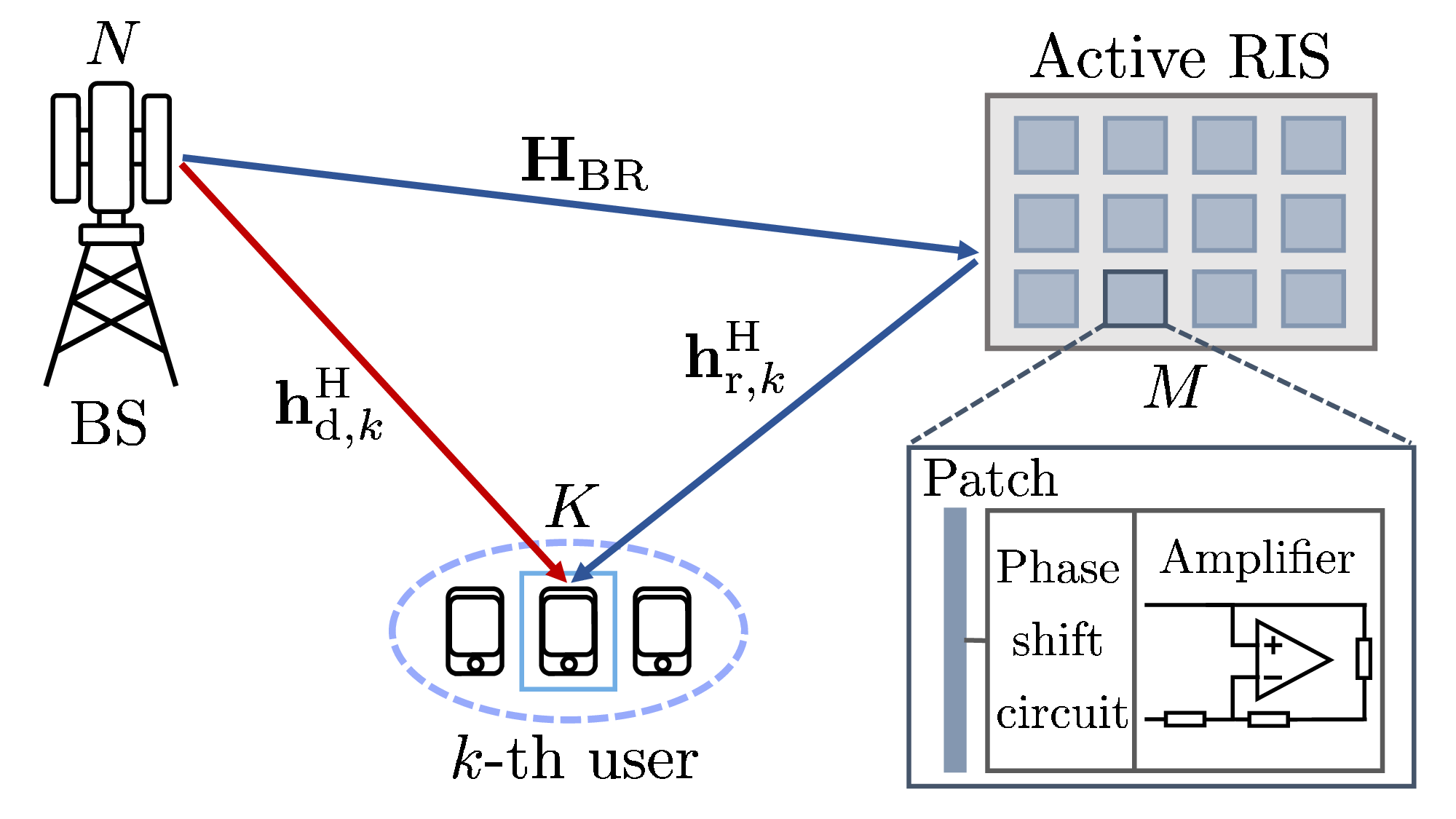}
    \caption{An active RIS-aided downlink system with $N$ BS antennas, $M$ RIS elements, and $K$ single-antenna users.} 
    \label{Fig:System_model}
\end{figure}
The transmit data symbol vector for $K$ users is $\bs = \left[ s_1, \ldots, s_K \right]^\mathrm{T} \in \mathbb{C}^{K \times 1}$ that satisfies $\mathbb{E} \left[ \bs\bs^\mathrm{H} \right] = \bI_K$.
The transmit beamformer is defined as $\bF = \left[ \bff_1, \ldots, \bff_K \right] \in \mathbb{C}^{N \times K}$, which satisfies $\left\Vert \bF \right\Vert_\mathrm{F}^2 \le P_\mathrm{B}$ with the downlink transmit power constraint $P_\mathrm{B}$.
The reflection coefficient matrix of the RIS is denoted by $\boldsymbol{\Lambda} = \operatorname{diag}(\bw) \in \mathbb{C}^{M \times M}$ with $\bw = \left[ w_1, \ldots, w_M \right]^\mathrm{T} \in \mathbb{C}^{M \times 1}$ where each element has a power limit, i.e., $|w_m|^2 \le a_{\mathrm{max}}$ for $m = 1, \cdots, M$, with the maximum amplification gain $a_{\mathrm{max}}$.
Then, the received signal at the $k$-th user is \cite{Active2, Active_Robust1}
\begin{equation}\label{received signal}
    y_k = \sum_{i=1}^K \left( \bh_{\mathrm{d}, k}^\mathrm{H} + \bh_{\mathrm{r}, k}^\mathrm{H} \boldsymbol{\Lambda} \bH_{\mathrm{BR}}\right) \bff_i s_i + \bh_{\mathrm{r}, k}^\mathrm{H} \boldsymbol{\Lambda} \bz + n_k,
\end{equation}
where $\bh_{\mathrm{d}, k}^\mathrm{H} \in \mathbb{C}^{1 \times N}, \bh_{\mathrm{r}, k}^\mathrm{H} \in \mathbb{C}^{1 \times M}$, and $\bH_{\mathrm{BR}} \in \mathbb{C}^{M \times N}$ denote the channels from the BS to the $k$-th user, from the RIS to the $k$-th user, and from the BS to the RIS, respectively.
The additive white Gaussian noise (AWGN) at the $k$-th user is defined as $n_k \sim \cC\cN \left( 0, \sigma_k^2 \right)$ with the noise variance $\sigma_k^2$.
Additionally, the AWGN at the RIS is modeled as $\bz \sim \cC\cN \left( \mathbf{0}_M, \sigma_\mathrm{z}^2 \bI_M \right)$ with the noise variance $\sigma_\mathrm{z}^2$, which is neglected at the passive RIS.
Since the thermal noise $\bz$ can be also amplified along with the incident signal, it cannot be ignored at the active RIS.
Due to deploying the active components in the reflective elements, the power consumption of the RIS is also considered as 
\begin{equation}
    \mathbb{E} \left[ \left\Vert \boldsymbol{\Lambda}   \left( \bH_{\mathrm{BR}} \bF \bs + \bz \right) \right\Vert_2^2 \right] = \left\Vert \boldsymbol{\Lambda} \bH_{\mathrm{BR}} \bF \right\Vert_\mathrm{F}^2 + \sigma_\mathrm{z}^2 \left\Vert \bw \right\Vert_2^2 \le P_\mathrm{R},
\end{equation}
with the RIS power constraint $P_\mathrm{R}$ assuming $\bs$ and all noises are independent, which is generally true in practice.

Due to the lack of baseband signal processing capability at the RIS, it is challenging to estimate the individual RIS-related channels accurately.
Accordingly, we suppose that there are estimation errors in the BS-user and the RIS-user channels.
For the $k$-th user, the channels are denoted by $\bh_{\mathrm{d}, k} = \widehat{\bh}_{\mathrm{d}, k} + \Delta \bh_{\mathrm{d}, k}$ and $\bh_{\mathrm{r}, k} = \widehat{\bh}_{\mathrm{r}, k} + \Delta \bh_{\mathrm{r}, k}$ where $\widehat{\bh}_{\mathrm{d}, k}$ and
$\widehat{\bh}_{\mathrm{r}, k}$ represent the estimated channels, and
$\Delta \bh_{\mathrm{d}, k}$ and $\Delta \bh_{\mathrm{r}, k}$ denote the estimation errors.
We adopt the statistical CSI error model \cite{Robust1, Robust2}, which follows the complex Gaussian distribution\footnote{The statistical characteristics of the CSI estimation error depend on the channel estimation method used. Specifically, when using the least squares or linear minimum mean square error (LMMSE) based estimation method, the estimation error follows the complex Gaussian distribution  \cite{Gaussian}.}, i.e., $\Delta \bh_{\mathrm{d}, k} \sim \cC\cN \left( \mathbf{0}_N, \sigma_{\mathrm{d},k}^2 \bI_N \right)$ and $\Delta \bh_{\mathrm{r}, k} \sim \cC\cN \left( \mathbf{0}_M, \sigma_{\mathrm{r}, k}^2 \bI_M \right)$ with the error variances $\sigma_{\mathrm{d},k}^2$ and $\sigma_{\mathrm{r}, k}^2$.
Note that we assume the BS can have accurate knowledge of $\bH_{\mathrm{BR}}$ \cite{Robust3, Robust4, Robust5}.
This is possible because both the BS and the RIS are typically installed at high locations, ensuring a line-of-sight path between them. In addition, since their positions are fixed and known in advance, the geometric parameters that constitute the channel can be accurately determined. Consequently, the BS-RIS channel can be estimated with greater precision than the other channels.

\section{Proposed Robust Transmission Design}
In this section, we propose a robust transmission design for active RIS-aided MU-MISO systems under the statistical CSI error model.
We first formulate an optimization problem to maximize the sum of the minimum achievable rate, which is obtained to address the estimation errors.
To tackle the non-concavity of the objective function, its lower bound is derived, which is concave with respect to each variable.
Then, we apply the AO method to solve the problem efficiently.

\subsection{Problem Formulation}
Since there are channel estimation errors, the sum of the achievable rate for each user cannot be directly used to optimize the beamformer at the BS and the reflection coefficients at the RIS.
Therefore, we aim to maximize the sum of the minimum achievable rate, $R_\mathrm{sum}$, which is derived to take the worst-case into account, by jointly optimizing the transmit beamformer at the BS and the reflection coefficients at the RIS.
The problem is formulated as follows
\begin{subequations}
    \begin{align}
        \left( \mathrm{P1} \right): \max_{\bF, \bw} ~& R_\mathrm{sum} = \sum_{k=1}^K R_{\mathrm{min}, k} \label{objective func.} \\ \mathrm{s.t.} \ & \Vert \bF \Vert_\mathrm{F}^2 \le P_\mathrm{B}, \\& \Vert \boldsymbol{\Lambda} \bH_{\mathrm{BR}} \bF \Vert_\mathrm{F}^2 + \sigma_\mathrm{z}^2 \Vert\bw\Vert_2^2 \le P_\mathrm{R}, \label{coupled} \\& \vert w_m \vert^2 \le a_{\mathrm{max}}, m = 1, \cdots, M,
    \end{align}
\end{subequations}
where $R_{\mathrm{min}, k}$ denotes the minimum achievable rate of the $k$-th user.
Under the statistical CSI error model, $R_{\mathrm{min}, k}$ is given~by
\begin{equation}\label{Min rate}
    R_{\mathrm{min}, k} = \log_2 \left( 1 + \frac{\vert \widehat{\bh}_k^\mathrm{H} \bff_k \vert^2}{\Gamma_k} \right),
\end{equation}
where $\widehat{\bh}_k^\mathrm{H}$ is the effective estimated channel from the BS to the $k$-th user, and $\Gamma_k$ is the interference-plus-noise considering the statistical characteristics of the estimation errors, which are
\begin{align}
&\widehat{\bh}_k^\mathrm{H} =  \widehat{\bh}_{\mathrm{d}, k}^\mathrm{H} +  \widehat{\bh}_{\mathrm{r}, k}^\mathrm{H} \boldsymbol{\Lambda} \bH_\mathrm{BR}, \\ &\Gamma_k = \sum_{i \neq k}^K \left\vert \widehat{\bh}_k^\mathrm{H} \bff_i \right\vert^2 + \sum_{i=1}^K ( \sigma_{\mathrm{d}, k}^2 \Vert\bff_i \Vert_2^2 + \sigma_{\mathrm{r}, k}^2 \Vert \boldsymbol{\Lambda} \bH_{\mathrm{BR}} \bff_i \Vert_2^2 ) \nonumber \\ & \qquad + \sigma_\mathrm{z}^2 \left\Vert \widehat{\bh}_{\mathrm{r},k}^\mathrm{H} \boldsymbol{\Lambda} \right\Vert_2^2 + \sigma_\mathrm{z}^2 \sigma_{\mathrm{r}, k}^2 \Vert \bw \Vert_2^2 + \sigma_k^2.
\end{align}
The detailed derivation of $R_{\mathrm{min}, k}$, including that of $\Gamma_k$, is provided in Appendix \ref{Appendix A}.
Note that the problem (P1) is non-convex since the objective function is not jointly concave with respect to $\bF$ and $\bw$, and the constraint \eqref{coupled} is non-convex due to the coupling of the variables.

\subsection{Problem Transformation} \label{Transform}
To handle the non-concave objective function of (P1), we first derive a tractable lower bound of $R_{\mathrm{min}, k}$ and utilize it to transform the problem (P1) into the sum of the lower bound maximization problem. Specifically, in an imperfect CSI scenario, we prove that the weighted MSE, where the expectation is taken with respect to the transmitted symbol, noises, and estimation errors, serves as a lower bound of the minimum achievable rate.

\begin{proposition} \label{prop.lower bound}
    For any $v_k \in \mathbb{R}_{++}$ and $u_k \in \mathbb{C}$, the following inequality holds
    \begin{equation}\label{lower bound}
        \log_2 \left( 1 + \frac{\vert \widehat{\bh}_k^\mathrm{H} \bff_k \vert^2}{\Gamma_k} \right) \ge \log_2 v_k - v_k e_k(u_k, \bF, \bw) + c,
    \end{equation}
    where $e_k(u_k, \bF, \bw) = \mathbb{E}\left[ (s_k - u_k^\mathrm{H} y_k)(s_k - u_k^\mathrm{H} y_k)^\mathrm{H} \right]$ is the MSE for the $k$-th user, and $c$ is constant, which are given by
    \begin{align}
        &e_k(u_k, \bF, \bw) = \left\vert 1 - u_k^\mathrm{H} \widehat{\bh}_k^\mathrm{H} \bff_k \right\vert^2 + \sum_{i \neq k}^K \left\vert u_k^\mathrm{H} \widehat{\bh}_k^\mathrm{H} \bff_i \right\vert^2 \nonumber \\& \qquad \ + \sum_{i=1}^K \left( \sigma_{\mathrm{d}, k}^2 \left\Vert u_k^\mathrm{H} \bff_i \right\Vert_2^2 + \sigma_{\mathrm{r}, k}^2 \left\Vert u_k^\mathrm{H} \boldsymbol{\Lambda} \bH_{\mathrm{BR}} \bff_i \right\Vert_2^2 \right) \nonumber \\& \qquad \ + \sigma_\mathrm{z}^2 \left\Vert u_k^\mathrm{H} \widehat{\bh}_{\mathrm{r},k}^\mathrm{H} \boldsymbol{\Lambda} \right\Vert_2^2 + \sigma_\mathrm{z}^2 \sigma_{\mathrm{r}, k}^2 \left\Vert u_k^\mathrm{H} \bw \right\Vert_2^2 + \sigma_k^2 \vert u_k \vert^2, \\ &c = \log_2 e - \log_2(\log_2 e).
    \end{align}
\end{proposition}
\begin{proof}
    Let us denote the function $f(v_k, u_k) = \log_2 v_k - v_k e_k(u_k, \bF, \bw) + c$ and the point $(v_k^\star, u_k^\star)$ that satisfies the first-order optimality condition, i.e.,
    \begin{equation}
        \frac{\partial f(v_k^\star, u_k^\star)}{\partial v_k} = 0, \    
        \frac{\partial f(v_k^\star, u_k^\star)}{\partial u_k} = 0.
    \end{equation}
    This point is computed as follows 
    \begin{equation}\label{opt.v and u}
        v_k^{\star} = \frac{1}{\ln2} \left( 1 + \frac{|\widehat{\bh}_k^\mathrm{H} \bff_k|^2}{\Gamma_k} \right), u_k^{\star} = \frac{\widehat{\bh}_k^\mathrm{H} \bff_k}{|\widehat{\bh}_k^\mathrm{H} \bff_k|^2 + \Gamma_k}.            
    \end{equation}
    By using the concavity of the function $f(v_k, u_k)$, we can show $f(v_k^\star, u_k^\star) \ge f(v_k, u_k^\star) \ge f(v_k, u_k)$.
    Applying \eqref{opt.v and u} gives $f(v_k^\star, u_k^\star) = \log_2 \left( 1 + \frac{|\widehat{\bh}_k^\mathrm{H} \bff_k|^2}{\Gamma_k} \right)$, which finishes the proof. 
\end{proof}

\noindent Note that, different from the conventional weighted MMSE method as in \cite{WMMSE}, we consider the CSI estimation errors and derive the relationship between $R_{\mathrm{min}, k}$ and the weighted MSE for each user.

By exploiting the lower bound derived in Proposition~\ref{prop.lower bound}, the problem (P1) can be transformed as
\begin{subequations}
    \begin{align}
        \left( \mathrm{P2} \right): \max_{\bv, \bu, \bF, \bw} &\sum_{k=1}^K  \log_2 v_k - v_k e_k(u_k, \bF, \bw) \label{new objec.} \\ \mathrm{s.t.} \ & \Vert \bF \Vert_\mathrm{F}^2 \le P_\mathrm{B}, \\& \Vert \boldsymbol{\Lambda} \bH_{\mathrm{BR}} \bF \Vert_\mathrm{F}^2 + \sigma_\mathrm{z}^2 \Vert\bw\Vert_2^2 \le P_\mathrm{R}, \\& \vert w_m \vert^2 \le a_{\mathrm{max}}, m = 1, \cdots, M,
    \end{align}    
\end{subequations}
where $\bv = \left[ v_1, \cdots, v_K \right]^\mathrm{T} \in \mathbb{R}_{++}^{K \times 1}$ and $\bu = \left[ u_1, \cdots, u_K \right]^\mathrm{T} \in \mathbb{C}^{K \times 1}$ denote the auxiliary variables.

\subsection{Alternating Optimization} \label{AO}
Although we considered the tractable lower bound, still the objective function in (P2) is not jointly concave with respect to $\bv$, $\bu$, $\bF$, and $\bw$.
However, it is concave with respect to each variable when the other variables are fixed; therefore, we adopt the AO method to obtain an effective solution of (P2).
Specifically, we first optimize the auxiliary variables $\bv$ and $\bu$ by fixing $\bF$ and $\bw$.
Since all constraints of (P2) are related to only $\bF$ and $\bw$, the problem (P2) is simplified as
\begin{align}
    \left( \mathrm{P2.1} \right): \max_{\bv, \bu} &\sum_{k=1}^K \log_2 v_k - v_k e_k(u_k, \bF, \bw),
\end{align}
where the optimal solutions $\bv^\star$ and $\bu^\star$ are obtained in the proof of Proposition~\ref{prop.lower bound}.

Then, we optimize the transmit beamformer $\bF$ while assuming the other variables are fixed.
By omitting the terms not related to $\bF$, the problem (P2) is represented as
\begin{subequations}
    \begin{align}
        \left( \mathrm{P2.2} \right): \min_{\bF} &\sum_{k=1}^K  v_k e_k(u_k, \bF, \bw) \label{new func. F} \\ \mathrm{s.t.} \ & \Vert \bF \Vert_\mathrm{F}^2 \le P_\mathrm{B}, \label{const1. F}\\ & \Vert \boldsymbol{\Lambda} \bH_{\mathrm{BR}} \bF \Vert_\mathrm{F}^2 \le P_\mathrm{R} -  \sigma_\mathrm{z}^2 \Vert\bw\Vert_2^2. \label{const2. F}
    \end{align}
\end{subequations}
Note that $e_k(u_k, \bF, \bw)$ can be expressed as the sum of the quadratic form of $\bff_k$ by expanding the square of the $\ell_2$-norms and the absolute values.
Therefore, $e_k(u_k, \bF, \bw)$ is convex with respect to each column of $\bF$, and by the definition of the convex function, it is convex with respect to $\bF$.
Since the objective function \eqref{new func. F} is convex, and all inequality constraints \eqref{const1. F} and \eqref{const2. F} imply convex feasible sets, the problem (P2.2) is the convex optimization problem whose optimal solution $\bF^\star$ can be obtained using the well-known tools such as CVX~\cite{CVX}.

Finally, we optimize the reflection coefficient vector $\bw$ while holding the other variables fixed.
The problem (P2) is expressed~as
\begin{subequations}
    \begin{align}
        \left( \mathrm{P2.3} \right): \min_{\bw} &\sum_{k=1}^K v_k e_k(u_k, \bF, \bw) \label{new func. w}\\
        \mathrm{s.t.} \ & \Vert \boldsymbol{\Lambda} \bH_{\mathrm{BR}} \bF \Vert_\mathrm{F}^2 + \sigma_\mathrm{z}^2 \Vert\bw\Vert_2^2 \le P_\mathrm{R}, \label{const1. w}\\& \vert w_m \vert^2 \le a_{\mathrm{max}}, m = 1, \cdots, M. \label{const2. w}
    \end{align}    
\end{subequations}
Since $e_k(u_k, \bF, \bw)$ can be rewritten in the quadratic form of $\bw$ by expanding the formulas and adjusting the order of the variables, it is convex with respect to $\bw$.
Hence, just as the objective function \eqref{new func. w} is convex with respect to $\bF$, the objective function is convex with respect to $\bw$, and all constraints \eqref{const1. w} and \eqref{const2. w} also imply convex feasible sets.
Therefore, the problem (P2.3) is the convex optimization problem that has the optimal solution~$\bw^\star$.

\subsection{Algorithm Description and Complexity Analysis}
The proposed robust design is summarized in Algorithm~\ref{Alg.1}.
After random initialization that satisfies all the power constraints, $\bv$, $\bu$, $\bF$, and $\bw$ are iteratively obtained by solving the problems (P2.1), (P2.2), and (P2.3).
The algorithm stops when the difference in the sum of the minimum achievable rate becomes less than a threshold $\epsilon$ or the iteration index $\ell$ reaches the maximum value $L$.
As the optimal solution for each variable is iteratively obtained, $R_\mathrm{sum}$ monotonically increases with each iteration.
Additionally, due to the power constraints at the BS and the RIS, $R_\mathrm{sum}$ is upper-bounded, which guarantees the convergence of Algorithm \ref{Alg.1}.

\alglanguage{pseudocode}
\begin{algorithm}[t] 
   \caption{Proposed robust transmission design for active RIS-aided MU-MISO systems} \label{Alg.1}
   \textbf{Initialize}
   \begin{algorithmic}[1]
      \State Set $L$, $ \epsilon > 0$, and $R_\mathrm{sum}^{(0)} = 0$
      \State Initialize $\bF^{(0)}$ and $\bw^{(0)}$ 
   \end{algorithmic}
   \textbf{Iterative update}
   \begin{algorithmic}[1]
      \addtocounter{ALG@line}{+2}
      \For{$\ell=1, \cdots, L$}
      \State Calculate $\bv^{\star(\ell)}$ and $\bu^{\star(\ell)}$ by \eqref{opt.v and u}
      \State Obtain $\bF^{\star(\ell)}$ from (P2.2)
      \State Obtain $\bw^{\star(\ell)}$ from (P2.3)
      \State Calculate $R_\mathrm{sum}^{(\ell)}$ by \eqref{Min rate}
      \If {$\left\vert R_\mathrm{sum}^{(\ell)}-R_\mathrm{sum}^{(\ell-1)} \right\vert < \epsilon$}
      \State Break
      \EndIf
      \EndFor
   \end{algorithmic}
\end{algorithm}

To analyze the complexity of the proposed algorithm, we compute the complexities of solving the problems (P2.1), (P2.2), and (P2.3), respectively. For (P2.1), the optimal solutions $\bv^\star$ and $\bu^\star$ are obtained by \eqref{opt.v and u} with the complexity $\mathcal{O} \left( M N K^2 \right)$. Then, since (P2.2) is the quadratically constrained quadratic program (QCQP) problem with $NK$ variables and two quadratic constraints, the required complexity to achieve $\bF^\star$ is $\mathcal{O} \left( N^3 K^3 \ln \frac{4 V}{\eta} \right)$ according to \cite{Complexity}. Note that $\eta$ denotes the solution accuracy and $V$ is a constant greater than or equal to the objective function and constraint functions. Similarly, (P2.3) is the QCQP problem with $M$ variables and $M+1$ quadratic constraints, and the required complexity to compute $\bw^\star$ is given by $\mathcal{O}\left( {M}^{3.5} \ln \frac{2 (M+1) V}{\eta} \right)$. Therefore, the overall maximum complexity of Algorithm 1 becomes $\mathcal{O}\left( L \left( M N K^2 + N^{3} K^{3} \ln \frac{4 V}{\eta} + {M}^{3.5} \ln \frac{2 (M+1) V}{\eta} \right) \right)$.

\section{Numerical Results}
In this section, we provide numerical results that demonstrate the performance of our proposed design.
We consider that the BS with a uniform planar array (UPA) of $N = 4 \times 4$ antennas serves $K = 4$ users.
The active RIS is equipped with a UPA of $M = 8 \times 8$ or $M = 12 \times 8$ elements.
The BS and the RIS are located at $(0 \ \mathrm{m}, 0 \ \mathrm{m})$ and $(500 \ \mathrm{m}, 20 \ \mathrm{m})$, and each user is randomly distributed in a circle centered at $(500 \ \mathrm{m}, 0 \ \mathrm{m})$ with a radius $10 \ \mathrm{m}$.
With the noise power spectral density $-174 \ \mathrm{dBm/Hz}$ and the bandwidth $10 \ \mathrm{MHz}$, the noise variances at the RIS and the $k$-th user are set to $\sigma_\mathrm{z}^2 = \sigma_k^2 = -104 \ \mathrm{dBm}$.
The RIS power constraint is set to $P_\mathrm{R} = -10 \ \mathrm{dBm}$ \cite{Active2}.

We assume the Rician fading for the estimated channels $\widehat{\bh}_{\mathrm{d}, k}$, $ \widehat{\bh}_{\mathrm{r}, k}$, and $\bH_{\mathrm{BR}}$ as described in \cite{Channel}, and the path-loss is set to $-30 \ \mathrm{dB}$ at the distance of $1 \ \mathrm{m}$.
The path-loss exponents of the BS-user, RIS-user, and BS-RIS channels are set to $4.0$, $2.3$, and $2.5$.
The CSI estimation error variances are denoted by $\sigma_{\mathrm{d}, k}^2 = \frac{\delta^2}{N} \Vert \widehat{\bh}_{\mathrm{d}, k} \Vert_2^2$ and $ \sigma_{\mathrm{r},k}^2 = \frac{\delta^2}{M} \Vert \widehat{\bh}_{\mathrm{r}, k} \Vert_2^2$ where $\delta \in [0, 1]$ represents the relative amount of CSI uncertainty.

To evaluate the performance of the proposed robust design, we adopt the sum-rate as the performance metric, which is given by 
\begin{align}\label{Sum-rate}
    \sum_{k=1}^K \log_2 \left( 1 + \frac{\vert \bh_k^\mathrm{H} \bff_k \vert^2}{\sum_{i \neq k}^K \vert \bh_k^\mathrm{H} \bff_i \vert^2 + \sigma_\mathrm{z}^2 \Vert \bh_{\mathrm{r},k}^\mathrm{H} \boldsymbol{\Lambda} \Vert_2^2 +  \sigma_k^2} \right),
\end{align}
where $\bh_k^\mathrm{H} = \bh_{\mathrm{d}, k}^\mathrm{H} + \bh_{\mathrm{r}, k}^\mathrm{H} \boldsymbol{\Lambda} \bH_{\mathrm{BR}}$ denotes the effective channel from the BS to the $k$-th user.
For performance comparison, the perfect CSI and the non-robust cases are considered as the benchmarks.
In the perfect CSI case, since all channels are perfectly known, the beamformer $\bF$ at the BS and the reflection coefficients $\bw$ at the RIS are jointly optimized to maximize the sum-rate in \eqref{Sum-rate}.
In the non-robust case, $\bF$ and $\bw$ are optimized as in the perfect CSI case using the estimated channels $\widehat{\bh}_{\mathrm{d}, k}$ and $\widehat{\bh}_{\mathrm{r}, k}$, not the perfect CSI.
For both baselines, we reformulate the problem by replacing the objective function with its lower bound and apply the AO method as in the proposed method.

Fig. \ref{Fig:Power_03} depicts the average sum-rate according to the downlink transmit power constraint $P_{\mathrm{B}}$ with the CSI uncertainty $\delta = 0.3$.
We set $a_{\mathrm{max}} = 30 \ \mathrm{dB}$ as the maximum amplification gain for each active RIS element.\footnote{Considering the practical hardware architecture of the reflective amplifier as in \cite{Amp_gain}, we assume that $a_\mathrm{max}$ can be set to less than or equal to $40 \ \mathrm{dB}$ for simulations.}
It can be observed that the proposed robust design outperforms the non-robust case, and the performance gap becomes larger as $P_\mathrm{B}$ increases.
These results imply that the considered worst-case maximization can bring noticeable performance gain under the presence of estimation errors.
In contrast, the non-robust case simply exploits only the estimated channels, leading to severe performance degradation.

\begin{figure}[t]
    \centering
    \includegraphics[width=0.7 \columnwidth]{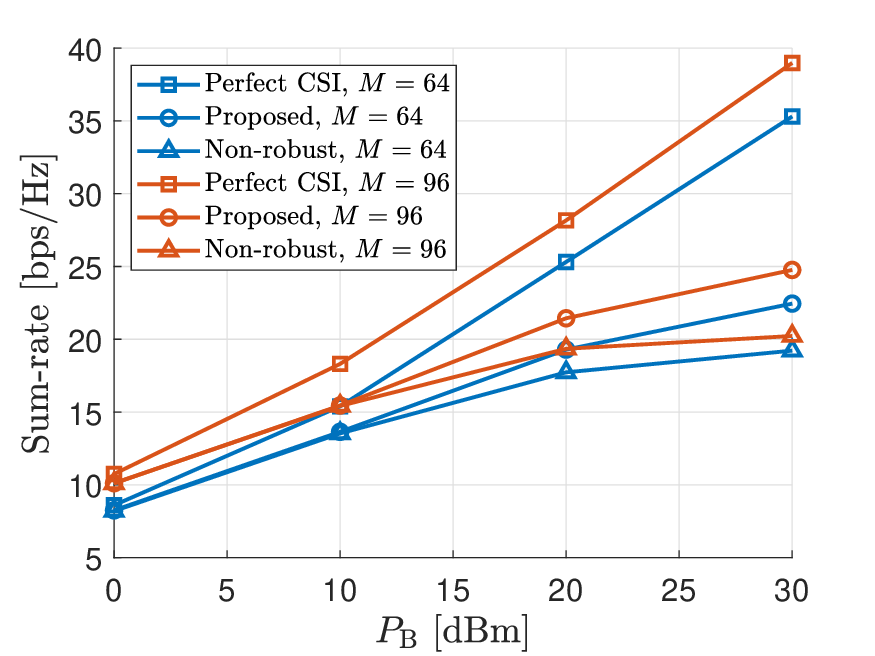}
    \caption{Average sum-rate performance according to $P_\mathrm{B}$ with $a_{\mathrm{max}} = 30 \ \mathrm{dB}$ and $\delta = 0.3$.} 
    \label{Fig:Power_03}
\end{figure}


Fig. \ref{Fig:Uncertainty} compares the average sum-rate according to the CSI uncertainty $\delta$ with $P_\mathrm{B} = 30 \ \mathrm{dBm}$ and $a_\mathrm{max} = 30 \ \mathrm{dB}$.
As $\delta$ increases, the CSI estimation error variances become larger, which makes the impact of channel mismatch more significant.
Therefore, the sum-rate performance of the proposed design and the non-robust case decrease.
Still, the proposed robust design shows higher sum-rate than the non-robust case by taking the estimation errors into account, and the performance gap becomes more pronounced for large $\delta$.
The sum-rate of the perfect CSI case slightly increases with $\delta$ because, according to our model explained in Section II, $\delta$ not only affects the CSI estimation error variances but also the average channel gains as $\bh_{\mathrm{d}, k} = \widehat{\bh}_{\mathrm{d}, k} + \Delta \bh_{\mathrm{d}, k}$ and $\bh_{\mathrm{r}, k} = \widehat{\bh}_{\mathrm{r}, k} + \Delta \bh_{\mathrm{r}, k}$.

\begin{figure}[t]
    \centering
    \includegraphics[width=0.7 \columnwidth]{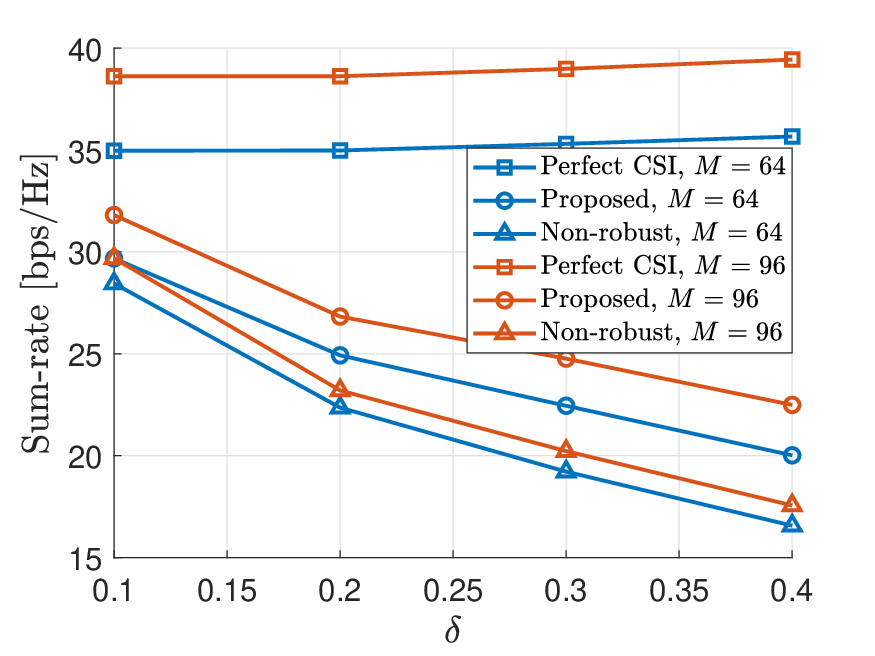}
    \caption{Average sum-rate performance according to $\delta$ with $P_\mathrm{B} = 30 \ \mathrm{dBm}$ and $a_{\mathrm{max}} = 30 \ \mathrm{dB}$.} 
    \label{Fig:Uncertainty}
\end{figure}

Fig.~\ref{Fig:a_max_03} shows the average sum-rate according to  $a_{\mathrm{max}}$ with $P_{\mathrm{B}} = 30 \ \mathrm{dBm}$ and $\delta = 0.3$.
The proposed robust design outperforms the non-robust case, especially in the high $a_\mathrm{max}$ regime.
This is because, for large $a_{\mathrm{max}}$, the estimation error of the RIS-user channel is further amplified in the active RIS, resulting in a larger mismatch of the effective channel.
Similar to Figs.~\ref{Fig:Power_03} and \ref{Fig:Uncertainty}, when the number of $M$ increases, the proposed design achieves much better sum-rate performance than the non-robust case.

\begin{figure}[t]
    \centering
    \includegraphics[width=0.7 \columnwidth]{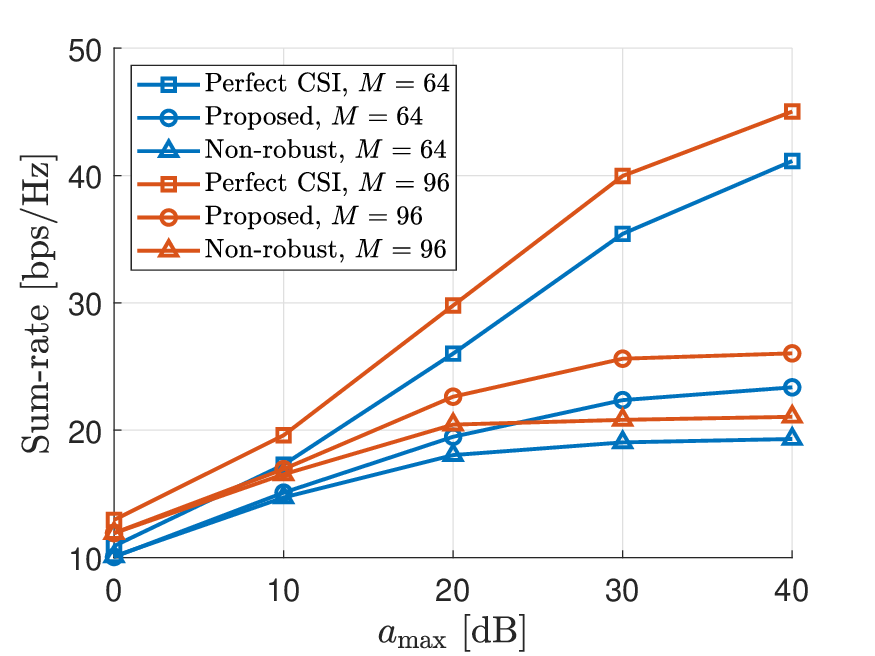}
    \caption{Average sum-rate performance according to $a_\mathrm{max}$ with $P_{\mathrm{B}} = 30 \ \mathrm{dBm}$ and $\delta = 0.3$.} 
    \label{Fig:a_max_03}
\end{figure}

\section{Conclusion}
In this paper, we investigated the robust transmission design for active RIS-aided MU-MISO systems with imperfect CSI.
We considered the worst-case scenario by deriving the minimum achievable rate of each user under the statistical CSI error model.
The problem was formulated to maximize the sum of the minimum achievable rate, and we reformulated the problem by replacing the objective function with the tractable lower bound.
Since this lower bound is concave with respect to each variable when the other variables are assumed to be fixed, we adopted the AO method to solve the problem.
Numerical results showed that the proposed robust design outperforms the non-robust case, especially when the channel uncertainty is large.
Possible future research directions can include developing the robust transmission design with a bounded CSI error model.
Additionally, it is worthwhile to investigate robust transmission designs that take into account the CSI estimation errors in the BS-RIS-user channels, or more generally, in the effective channels.

\begin{appendices} 
\section{Derivation of $R_{\mathrm{min}, k}$} \label{Appendix A}
For notational simplicity, let us define the set that contains all channel information known at the BS as $\mathbb{S}_{\mathrm{Ch}} \triangleq \lbrace \lbrace \widehat{\bh}_{\mathrm{d}, k}, \widehat{\bh}_{\mathrm{r}, k} \rbrace_{k=1}^K, \bH_\mathrm{BR} \rbrace$. 
At the $k$-th user, the conditional mutual information between the transmitted symbol and the received signal is defined as $\mathcal{I} \left(s_k ; y_k \vert  \mathbb{S}_{\mathrm{Ch}} \right)$, which can be represented by differential entropies as 
\begin{equation}
    \mathcal{I} \left( s_k ; y_k \vert  \mathbb{S}_{\mathrm{Ch}} \right) = \mathcal{H} \left( s_k \vert  \mathbb{S}_{\mathrm{Ch}} \right) - \mathcal{H} \left( s_k \vert y_k, \mathbb{S}_{\mathrm{Ch}} \right).
\end{equation}
The first term $\mathcal{H} \left( s_k \vert  \mathbb{S}_{\mathrm{Ch}} \right)$ is simplified to $\log_2 \pi e$ by assuming $s_k$ as the Gaussian symbol.
Then, we can derive an upper bound of the second term as \cite{Appendix}
\begin{align}\label{upper bound}
    \mathcal{H} \left( s_k \vert y_k, \mathbb{S}_{\mathrm{Ch}} \right) & \le \log_2 \left( \pi e \mathbb{E}\left[ (s_k - \alpha_k y_k)(s_k - \alpha_k y_k)^\mathrm{H} \right] \right) \nonumber \\ &= \log_2 \left( \pi e \left( 1 - \frac{\mathbb{E}\left[ s_k y_k^\mathrm{H} \right] \mathbb{E}\left[ s_k^\mathrm{H} y_k \right] }{\mathbb{E}\left[ y_k y_k^\mathrm{H} \right]} \right) \right),
\end{align}
where $\alpha_k$ is set to the MMSE weight. To compute $\mathbb{E}\left[ s_k y_k^\mathrm{H} \right]$ and $\mathbb{E}\left[ y_k y_k^\mathrm{H} \right]$, different from \cite{Appendix}, we need to take into account our system model including the CSI error model. Then, the received signal $y_k$ in~\eqref{received signal} can be expressed as
\begin{equation}
    y_k = \sum_{i=1}^K \left( \widehat{\bh}_k^\mathrm{H} \bff_i s_i + \Delta \bh_k^\mathrm{H} \bff_i s_i \right) + \widehat{\bh}_{\mathrm{r}, k}^\mathrm{H} \boldsymbol{\Lambda} \bz + \Delta \bh_{\mathrm{r}, k}^\mathrm{H} \boldsymbol{\Lambda} \bz + n_k,
\end{equation}
where $\Delta \bh_k^\mathrm{H} = \Delta \bh_{\mathrm{d}, k}^\mathrm{H} + \Delta \bh_{\mathrm{r}, k}^\mathrm{H} \boldsymbol{\Lambda} \bH_{\mathrm{BR}}$ denotes the estimation error of the effective channel from the BS to the $k$-th user, which follows the complex Gaussian distribution as $\Delta \bh_k \sim \cC\cN \left( \mathbf{0}_N, \sigma_{\mathrm{d}, k}^2 \bI_N + \sigma_{\mathrm{r}, k}^2 \bH_{\mathrm{BR}}^\mathrm{H} \boldsymbol{\Lambda}^\mathrm{H} \boldsymbol{\Lambda} \bH_{\mathrm{BR}} \right)$. By leveraging the fact that the transmitted symbols, noises, and estimation errors are uncorrelated and have zero mean, we can obtain the following results
\begin{align} \label{Correlation}
    &\mathbb{E} \left[ s_k y_k^\mathrm{H} \right] = \bff_k^\mathrm{H} \widehat{\bh}_k, \\& \mathbb{E} \left[ y_k y_k^\mathrm{H} \right] = \sum_{i=1}^K \left\vert \widehat{\bh}_k^\mathrm{H} \bff_i \right\vert^2 + \sum_{i=1}^K \underbrace{\mathbb{E} \left[ \Delta \bh_k^\mathrm{H} \bff_i s_i s_i^\mathrm{H} \bff_i^\mathrm{H} \Delta \bh_k \right]}_{\mathrm{(a)}} \nonumber \\& \quad + \sigma_z^2 \left\Vert \widehat{\bh}_{\mathrm{r},k}^\mathrm{H} \boldsymbol{\Lambda} \right\Vert_2^2 + \underbrace{\mathbb{E} \left[ \Delta \bh_{\mathrm{r}, k}^\mathrm{H} \boldsymbol{\Lambda} \bz \bz^\mathrm{H} \boldsymbol{\Lambda}^\mathrm{H} \Delta \bh_{\mathrm{r}, k} \right]}_{\mathrm{(b)}} + \sigma_k^2. \label{y_k variance}
\end{align}
To achieve closed-form expressions of (a) and (b), the cyclic property of the trace operation is used as follows
\begin{align}
    &\mathbb{E} \left[ \Delta \bh_k^\mathrm{H} \bff_i s_i s_i^\mathrm{H} \bff_i^\mathrm{H} \Delta \bh_k \right] \nonumber \\ &= \mathbb{E} \left[ \operatorname{Tr} \left( \Delta \bh_k \Delta  \bh_k^\mathrm{H} \bff_i s_i s_i^\mathrm{H} \bff_i^\mathrm{H} \right) \right] \nonumber \\ &= \operatorname{Tr} \left( \left( \sigma_{\mathrm{d}, k}^2 \bI_N + \sigma_{\mathrm{r}, k}^2 \bH_{\mathrm{BR}}^\mathrm{H} \boldsymbol{\Lambda}^\mathrm{H} \boldsymbol{\Lambda} \bH_{\mathrm{BR}} \right) \bff_i \bff_i^\mathrm{H} \right) \nonumber \\ &= \operatorname{Tr} \left( \sigma_{\mathrm{d}, k}^2 \bff_i \bff_i^\mathrm{H} \right) + \operatorname{Tr} \left( \sigma_{\mathrm{r}, k}^2 \bff_i^\mathrm{H} \bH_{\mathrm{BR}}^\mathrm{H} \boldsymbol{\Lambda}^\mathrm{H} \boldsymbol{\Lambda} \bH_{\mathrm{BR}} \bff_i \right) \nonumber \\ &= \sigma_{\mathrm{d}, k}^2 ||\bff_i||_2^2 + \sigma_{\mathrm{r}, k}^2 ||\boldsymbol{\Lambda} \bH_{\mathrm{BR}} \bff_i||_2^2, \label{expression (a)} \\ &\mathbb{E} \left[ \Delta \bh_{\mathrm{r}, k}^\mathrm{H} \boldsymbol{\Lambda} \bz \bz^\mathrm{H} \boldsymbol{\Lambda}^\mathrm{H} \Delta \bh_{\mathrm{r}, k} \right] \nonumber \\ &= \mathbb{E} \left[ \operatorname{Tr} \left(\Delta \bh_{\mathrm{r}, k} \Delta \bh_{\mathrm{r}, k}^\mathrm{H} \boldsymbol{\Lambda} \bz \bz^\mathrm{H} \boldsymbol{\Lambda}^\mathrm{H} \right) \right] \nonumber \\& = \operatorname{Tr} \left( \sigma_z^2 \sigma_{\mathrm{r}, k}^2 \boldsymbol{\Lambda} \boldsymbol{\Lambda}^\mathrm{H} \right) \nonumber \\& = \sigma_z^2 \sigma_{\mathrm{r},k}^2 ||\bw||_2^2. \label{expression (b)}
\end{align}
By adopting \eqref{expression (a)} and \eqref{expression (b)} directly, \eqref{y_k variance} can be expressed as
\begin{align}\label{Variance}
    &\mathbb{E} \left[ y_k y_k^\mathrm{H} \right] = \vert \widehat{\bh}_k^\mathrm{H} \bff_k \vert^2 + \underbracea{\sum_{i \neq k}^K \left\vert \widehat{\bh}_k^\mathrm{H} \bff_i \right\vert^2 + \sum_{i=1}^K ( \sigma_{\mathrm{d}, k}^2 \Vert\bff_i \Vert_2^2 } \nonumber \\& \underbraced{ + \sigma_{\mathrm{r}, k}^2 \Vert \boldsymbol{\Lambda} \bH_{\mathrm{BR}} \bff_i \Vert_2^2 ) + \sigma_\mathrm{z}^2 \left\Vert \widehat{\bh}_{\mathrm{r},k}^\mathrm{H} \boldsymbol{\Lambda} \right\Vert_2^2 + \sigma_\mathrm{z}^2 \sigma_{\mathrm{r}, k}^2 \Vert \bw \Vert_2^2 + \sigma_k^2.}_{\triangleq \Gamma_k}
\end{align}
By applying \eqref{Correlation} and \eqref{Variance}, the upper bound of $\mathcal{H} \left( s_k \vert y_k, \mathbb{S}_{\mathrm{Ch}} \right)$ is computed as 
\begin{align}
    \mathcal{H} \left( s_k \vert y_k, \mathbb{S}_{\mathrm{Ch}} \right) \nonumber &\le \log_2 \pi e \left( \frac{\Gamma_k}{\vert \widehat{\bh}_k^\mathrm{H} \bff_k \vert^2 + \Gamma_k} \right) \\ &= \log_2 \pi e - \log_2 \left( 1 + \frac{\vert\widehat{\bh}_k^\mathrm{H} \bff_k \vert^2}{\Gamma_k} \right).
\end{align}
Finally, the lower bound of the conditional mutual information is obtained as
\begin{equation}
    \mathcal{I} \left( s_k ; y_k \vert \mathbb{S}_{\mathrm{Ch}} \right) \ge \log_2 \left( 1 + \frac{\vert\widehat{\bh}_k^\mathrm{H} \bff_k \vert^2}{\Gamma_k} \right).
\end{equation}

\end{appendices}

\bibliographystyle{IEEEtran}
\bibliography{refs_all}

\end{document}